\documentclass[11pt]{article}
\usepackage{amsthm,amsfonts,amsmath,amssymb}
\usepackage{graphicx}
\theoremstyle{plain}

\newtheorem{thm}{Theorem}[section]
\newtheorem{lem}[thm]{Lemma}
\newtheorem{prop}[thm]{Proposition}

\theoremstyle{definition}

\newtheorem{rem}[thm]{Remark}
\theoremstyle{remark}

\def\F{{\cal F}}
\def\Fh{\hat{\cal F}}

\newcommand{\PRE}{Phys. Rev. E.}
\newcommand{\PRL}{Phys. Rev. Lett.}

\newcommand{\be}{\begin{equation}}
\newcommand{\ee}{\end{equation}}
\newcommand{\bea}{\begin{eqnarray}}
\newcommand{\eea}{\end{eqnarray}}

\renewcommand{\atop}[2]{\genfrac{}{}{0pt}{}{#1}{#2}}

\newcommand{\cthm}[1]{Theorem~\ref{#1}}
\newcommand{\cprop}[1]{Proposition~\ref{#1}}
\newcommand{\clem}[1]{Lemma~\ref{#1}}

\newcommand{\crem}[1]{Remark~\ref{#1}}

\newcommand{\E}{{\mathcal E}}
\newcommand{\zetau}{{\underline\zeta}}
\newcommand{\xiu}{{\underline\xi}}
\renewcommand{\S}{{\mathcal S}}

\newcommand{\mbond}{\hbox{---}\llap{$\times\, $}\ }
\newcommand{\bond}{\hbox{---}\ }

\newcommand\bN 
{{\mathbb N}}

\newcommand\bZ{{\mathbb Z}}
\newenvironment{proofof}[1]{\medskip\noindent
   \textit{Proof of #1:} }{\hfill \qed\par\medskip}
\newenvironment{proofof*}[1]{\medskip\noindent
   \textit{Proof of #1:} }{}

\numberwithin{equation}{section} 
\def\line{\hbox to\hsize}

\title{The grand canonical ABC model: a reflection asymmetric  mean field
 Potts model}
\author{{ J. Barton$^{1}$, J. L. Lebowitz${}^{1,2}$,
and E. R. Speer${}^2$}\\ \\
{\small $^1$ Department of Physics, Rutgers University,}\\
{\small Piscataway, NJ 08854 USA}\\ 
{\small $^2$ Department of Mathematics, Rutgers University,}\\
{\small Piscataway, NJ 08854-8019 USA}}

\begin{document}

\maketitle

\begin{abstract} We investigate the phase diagram of a three-component
system of particles on a one-dimensional filled lattice, or equivalently of
a one-dimensional three-state Potts model, with reflection asymmetric mean
field interactions.  The three types of particles are designated as $A$,
$B$, and $C$.  The system is described by a grand canonical ensemble with
temperature $T$ and chemical potentials $T\lambda_A$, $T\lambda_B$, and
$T\lambda_C$.  We find that for $\lambda_A=\lambda_B=\lambda_C$ the system
undergoes a phase transition from a uniform density to a continuum of
phases at a critical temperature $\hat T_c=(2\pi/\sqrt3)^{-1}$.  For other
values of the chemical potentials the system has a unique equilibrium
state.  As is the case for the canonical ensemble for this $ABC$ model, the
grand canonical ensemble is the stationary measure satisfying detailed
balance for a natural dynamics.  We note that $\hat T_c=3T_c$, where $T_c$
is the critical temperature for a similar transition in the canonical
ensemble at fixed equal densities $r_A=r_B=r_C=1/3$. \end{abstract}


\section{Introduction\label{sec:intro}}

In this paper we study the phase diagram of the three species ABC model on
an interval as a function of the chemical potentials and the temperature.
The system is defined microscopically on a lattice of $N$ sites in which
each site is occupied by either an $A$, a $B$, or a $C$ particle.  
The energy is of mean field type, with an
interaction which has cyclic symmetry in the particle types but is
reflection asymmetric:
\be
E_N(\zetau)  = {\frac{1}{N}} \sum^{N-1}_{i=1}
\sum^{N}_{j=i+1} [ \eta_C (i) \eta_B(j) + \eta_A (i) \eta_C(j) +
\eta_B(i) \eta_A (j)]. \label{eq:EN}
\ee
 Here the configuration $\zetau$ of the model is an $N$-tuple
$(\zeta_1,\ldots,\zeta_N)$, with $\zeta_i=A$, $B$, or $C$, and
$\eta_\alpha(i)$, $\alpha=A,B,C$, is a random variable which specifies
whether a particle of species $\alpha$ is present at site $i$:
$\eta_\alpha(i)=1$ if $\zeta_i=\alpha$ and $\eta_\alpha(i)=0$ otherwise, so
that always $\sum_\alpha\eta_i(\alpha)=1$.  

We remark that we may also regard the model as a reflection asymmetric mean
field three state Potts model.  The asymmetry of the interaction, however,
gives this model system very different behavior from that of the usual
symmetric mean field model \cite{SF}.  Similar but short-range (in fact,
nearest neighbor) reflection asymmetric interactions occur in chiral clock
models \cite{O,H,AP}; see Remark~\ref{rem:chiral} below.

The equilibrium probability of a configuration $\zetau$ 
is given by the grand canonical Gibbs measure
\be
 \mu_{\beta,\lambda}(\zetau) 
  = \Xi^{-1} \exp [-\beta E_N(\zetau)
  +\sum_\alpha \lambda_\alpha N_\alpha(\zeta)], \label{eq:mu}
\ee
 where $\beta$ is the inverse temperature, $\lambda_A$, $\lambda_B$, and
$\lambda_C$ are $\beta$ times the chemical potentials,
$N_\alpha=\sum^N_{i=1} \eta_{\alpha}(i)$ with $\sum_\alpha N_\alpha=N$, and
$\Xi$ is the usual grand canonical partition function.  We prove here that
in the scaling limit ($N\to\infty$, $i/N\to x\in[0,1]$) the equilibrium
density profiles $\rho(x)$ are unique and spatially nonuniform when the
$\lambda_\alpha$'s are not all the same.  When
$\lambda_A=\lambda_B=\lambda_C$ the densities are spatially uniform above a
critical temperature $\hat T_c=\hat\beta_c^{-1}$, with
$\hat\beta_c=2\pi/\sqrt3$; below $\hat T_c$ the profiles have a natural
extension to periodic functions with a period three times the length of the
system.

One may compare the behavior described above with that of the same system
in the canonical ensemble, in which the $N_\alpha$ are taken as fixed; this
is the only case considered previously.  The results are quite different,
that is, we have inequivalence of ensembles (see \cite{CDR,M} for recent
reviews).  We give in Section~2 a brief history of the ABC model with fixed
particle number and a summary of results for that system.  In Section~3 we
describe a stochastic evolution satisfying detailed balance with respect to
the measures $\mu_{\beta,\lambda}(\zetau)$, and in Section~4 we establish
the phase diagram.  Section~5 gives a discussion of some related models and
problems.

\section{The ABC model in the canonical ensemble\label{sec:canonical}}

The ABC model was introduced by Evans et al.~\cite{Evans98} (see also
\cite{EKKM2,CDE,BDLvW,BDGJ,FF,FF2,ACLMMS}) as a one dimensional system
consisting of three species of particles, labeled $A, B, C$, on a ring
containing $N$ lattice sites; we will typically let $\alpha=A$, $B$, or $C$
denote a particle type, and make the convention that $\alpha+1$,
$\alpha+2$, $\ldots$ denote the particle types which are successors to
$\alpha$ in the cyclic order $ABC$.  The system evolves by nearest neighbor
exchanges with asymmetric rates: if sites $i$ and $i+1$ are occupied by
particles of different types $\alpha$ and $\gamma$, respectively, then the
exchange $\alpha\,\gamma\to\gamma\,\alpha$ occurs at rate $q<1$ if
$\gamma=\alpha+1$ and at rate $1$ if $\gamma=\alpha-1$.  The total numbers
$N_\alpha$ of particles of each species are conserved and satisfy
$\sum_\alpha N_\alpha=N$.  In the limit $N \to \infty$ with
$N_\alpha/N\to r_\alpha$, where $r_\alpha>0$ for all $\alpha$, the system
segregates into pure $A$, $B$, and $C$ regions, with rotationally invariant
distribution of the phase boundaries.

In the weakly asymmetric version of the system introduced by Clincy et
al.~\cite{CDE}, in which $q = e^{-\beta/N}$, the stationary state for the
equal density case $N_A=N_B=N_C$ is a Gibbs measure of the form
$\exp \{-\beta E_N\}$, so that the parameter $\beta=T^{-1}$ plays the
role of an inverse temperature.  The energy $E_N$ is given by
\eqref{eq:EN}, and the condition $N_A = N_B = N_C$ ensures that this is
translation invariant, despite the appearance of a preferred starting site
for the summations.

Ayyer et.~al.~\cite{ACLMMS} studied the weakly asymmetric system on an
interval, that is, again on a one-dimensional lattice of $N$ sites but now
with zero flux boundary conditions, so that a particle at site $i=1$
(respectively $i=N$) can only jump to the right (respectively left).  For
this system the steady state is {\it always} Gibbsian, given by
$\exp\{-\beta E_N\}$ with $E_N$ as in \eqref{eq:EN}, whatever the values of
$N_A$, $N_B$, and $N_C$.  When $N_A=N_B=N_C$ the steady state of the system
thus agrees with that on the ring, so that the invariance under rotations
on the ring then implies a rather surprising ``rotation'' invariance of the
Gibbs state on the interval.  We describe the results of \cite{ACLMMS} in
some detail, since the work of the  current paper depends heavily on
them. 

 To identify typical coarse-grained density profiles at large $N$,
\cite{ACLMMS} considers the scaling limit
 \be\label{eq:scale}
   N\to\infty,\qquad i/N\to x,\quad x\in[0,1]. 
 \ee
 For this limit there exists a Helmholtz free energy functional
$\beta^{-1}\F(\{n\})$ of the density profile $n(x)=(n_A(x),n_B(x),n_C(x))$.
$\F$ is the difference of contributions from the energy and entropy:
 \be
 \F(\{n\})=\beta \E(\{n\}) -\S(\{n\}), \label{eq:F}
 \ee
 where $\E\left(\{n\}\right) $ and $\S\left(\{n\}\right) $ are the limiting
values of the energy and entropy per site:
 \bea
\E\left(\{n\}\right) &=&\int_0^1 \,dx
\int_x^1\,dz  \sum_\alpha n_\alpha(x)n_{\alpha+2}(z),\label{eq:E}\\
\S(\{n\})&=&-\int_0^1 dx \sum_\alpha n_{\alpha}(x) \ln n_{\alpha}(x).
\label{eq:S}
 \eea
 We will write $\F=\F^{(\beta)}$ when we need to indicate explicitly the
$\beta$ dependence.  Only the canonical ensemble was considered in
\cite{ACLMMS}, so that for some fixed positive mean densities $r_A$, $r_B$,
$r_C$ satisfying $r_A+r_B+r_C=1$ the profiles $n(x)$ in
\eqref{eq:F}--\eqref{eq:S} satisfy the conditions
 \be   \label{eq:constraints}
0\le n_\alpha(x)\le1,\quad \sum_\alpha n_\alpha(x) = 1,
  \quad\hbox{and}\quad
  \int^1_0 n_\alpha(x) dx = r_\alpha.
 \ee
The typical profiles in the scaling limit are those which
minimize $\F$; it was shown in \cite{ACLMMS} that such minimizers always
exist and satisfy the ELE derived from $\F$. To obtain the ELE one defines
\begin{align}\label{eq:FA}
  \F_\alpha(x) &= \frac{\delta\F}{\delta n_\alpha(x)}\nonumber\\
  &=\log n_\alpha(x) 
   + \beta \int_0^x[n_{\alpha+1}(z)-n_{\alpha+2}(z)]\,dz+1+\beta r_{\alpha+2}
\end{align}
 to be the variational derivative taken as if the profiles $n_A(x)$,
$n_B(x)$, and $n_C(x)$ were independent; the constraints
\eqref{eq:constraints} then imply that at a stationary point of $\F$ both
$\F_A-\F_C$ and $\F_B-\F_C$ are constant.  After simple manipulations (see
also Section~\ref{sec:scaling} below) this yields the ELE satisfied by the
typical profiles $\rho(x)$:
  \be\label{eq:ABC}
 \frac{d\rho_\alpha}{dx} =  
  \beta\rho_\alpha ( \rho_{\alpha-1}- \rho_{\alpha+1}),\qquad \alpha=A,B,C.
 \ee
 These are to be solved subject to \eqref{eq:constraints} (written in terms
of $\rho$ rather than $n$).

It follows from \eqref{eq:ABC} that all relevant solutions satisfy
$\prod_\alpha\rho_\alpha(x)=K$ for some constant $K$ with $0<K\le 1/27$.
For $K=1/27$ they are constant, with value $1/3$; for $K<1/27$ they have
the form
 \be\label{eq:solns}
\rho_\alpha(x)=y_K(2\beta(x-1/2) + t_\alpha), \qquad 0\le x\le 1,
 \ee
 with $y_K(t)$ a solution, periodic with period $\tau_K$, of the equation
 \be\label{eq:osc}
 \frac12{y'}^2+\frac12Ky-\frac18y^2(1-y)^2=0;
 \ee
  here $t=2\beta x+\rm constant$. $y_K$ is uniquely specified by requiring
that it take on its minimum value at the points $t=n\tau_K$, $n\in\bZ$.
The phase shifts $t_\alpha$ in \eqref{eq:solns} satisfy
 \be\label{eq:phases}
   t_A=t_B+\tau_K/3 \qquad\hbox{and}\qquad t_C=t_B-\tau_K/3.
 \ee

\begin{rem} \label{rem:yk} Equation \eqref{eq:osc} describes a particle of
unit mass and zero energy oscillating in a potential
$U_K(y)=Ky/2-y^2(1-y)^2/8$.  The constant solution $y=1/3$ appears for
$K=1/27$.  For $K<1/27$, $y_K(t)$ is an even function which is strictly
increasing on the interval $[0,\tau_K/2]$; it was shown in \cite{ACLMMS}
that $\tau_K$ is a strictly decreasing function of $K$.  Because the
potential is quartic in $y$ the solution is an elliptic function.  Further
properties of the function $y_K$ are summarized in \cprop{recall} of
Appendix~\ref{sec:proof}.  \end{rem}

Equation \eqref{eq:solns} indicates that nonconstant solutions of the ELE
are obtained by viewing $y_K(t)$, and its translates by one-third and
two-thirds of a period, in some ``window'' of length $2\beta$.  If one is
given $\beta$ and $r=(r_A,r_B,r_C)$ then one must determine $K$ and one of
the phase shifts, say $t_B$, so that
 \be\label{eq:ra}
r_\alpha=\frac1{2\beta}\int_{-\beta}^\beta
y_K(t+t_\alpha)\,dt,\qquad \alpha=A,B,C. 
 \ee
  The solutions which minimize $\F$ were completely determined in
\cite{ACLMMS}.  In stating the result, we use the following terminology: a
solution is {\it of type $n$} if $(n-1)\tau_K<2\beta\le n\tau_K$, that is,
if the window contains more than $n-1$ and at most $n$ periods of the
function $y_K$.

 \begin{thm}\label{thm:oldmain} Suppose that $r_A$, $r_B$, and $r_C$ are
strictly positive. Then:

\smallskip\noindent
  (a) If $r_A=r_B=r_C=1/3$ then for the equations \eqref{eq:ABC} with
\eqref{eq:constraints} there exist (i)~the constant solution, (ii)~for
$\beta>n\beta_c=2\pi n\sqrt3$, $n=1,2,\ldots$, a family of solutions, of
period $\tau_K=2\beta/n$ and hence of type $n$, differing by translation,
and (iii)~no other solutions.  The minimizers of the free energy are, for
$\beta\le\beta_c$, the (unique) constant solution and, for $\beta>\beta_c$,
any type~1 solution.

\smallskip\noindent
 (b) For values of $r$ other than $(1/3,1/3,1/3)$ there exists for all
$\beta$ a unique type~1 solution of these equations which is a minimizer of
the free energy. 

\smallskip\noindent (c) At zero temperature ($\beta\to\infty$) the system
segregates into either three or four blocks, each containing 
particles of only one type. \end{thm}

\section{Dynamics of the grand canonical ABC model\label{sec:dynamics}}

We now turn to consideration of the ABC model on the interval when the
number of particles can fluctuate; we will abbreviate this as the GCABC
model.  In Section~1 the corresponding grand canonical measure
$\mu_{\beta,\lambda}$ (see \eqref{eq:mu}) was presented in the equilibrium
setting as a Gibbs measure obtained from the energy function \eqref{eq:EN}
and chemical potentials $\beta^{-1}\lambda_\alpha$.  Just as for the
canonical Gibbs measure, however, one may alternatively view this as the
stationary measure for some dynamics; we describe two possibilities here.
(A different generalization of the ABC model to a nonconserving dynamics,
in which the system is on a ring, vacancies are permitted, and the total
number of particles fluctuates but the differences $N_\alpha-N_\gamma$ are
conserved, is given in \cite{LM,LCM}.  When all the $N_\alpha$ are equal
the stationary measure has the form of a grand canonical ensemble.)

In the first dynamics we consider there are particle exchanges between
adjacent sites, with the same rates as for the canonical dynamics.  To
allow the number of particles to fluctuate, however, we introduce two new
possible transitions.  First, if the particle at site $i=1$ is of type
$\alpha$ then with a rate equal to $Ce^{-\lambda_\alpha}$ the entire
configuration is translated by one site to the left, the particle at site
$i=1$ disappears, and a particle of species $\alpha+1$ is created at site
$i=N$. Second, with a rate equal to $Ce^{-\lambda_{\alpha+1}}$ the reverse
transition occurs.  Here $C$ is a constant which we shall in the future
take equal to~1.  This dynamics satisfies the detailed balance condition
with respect to the Gibbs measure \eqref{eq:mu}: if a transition
$\zetau\to\zetau'$ arises from an exchange of particles the argument is as
for the canonical model \cite{ACLMMS}, while if it comes from a transition
of the new type, say in the ``forward'' direction as described above, then
$E_N(\zetau)=E_N(\zetau')$ but $N_\alpha$ decreases by 1 and $N_{\alpha+1}$
increases by 1, and the detailed balance condition
$e^{-\lambda_\alpha}\nu_\beta(\zetau)
=e^{-\lambda_{\alpha+1}}\nu_\beta(\zetau')$ follows.

\begin{rem}
 One may also obtain this dynamics  by considering a ring
of $N$ sites, with each site occupied by an $A$, $B$, or $C$ particle and
with a marker located on one of the bonds between adjacent sites.  Adjacent
particles exchange across any unmarked bond with the usual ABC rates, while
the marker may move one bond to its left or right, and in doing so it
changes the species of the particle it passes: with \mbond and \bond
denoting a marked and unmarked bond, respectively, the transition
$\mbond\alpha\ \bond\to\bond(\alpha+1)\ \mbond$ occurs with a rate equal to
$e^{-\lambda_\alpha}$ and the reverse transition with a rate equal to
$e^{-\lambda_{\alpha+1}}$.  If one then obtains a configuration on the
interval from a ring configuration by letting the marked bond identify the
boundaries of the interval---effectively by cutting the ring at the marked
bond---one sees easily that the inherited dynamics on the interval is
precisely the dynamics discussed above.  A
slight variation of this idea was mentioned in \cite{ACLMMS}.
\end{rem}

We define the second dynamics only for the case in which all the
$\lambda_\alpha$ are equal.  We obtain it by first defining a dynamics for
the {\it constrained ring:} a ring of $3N$ sites populated by $A$, $B$, and
$C$ particles but with a restriction to configurations $(\xi_i)_{i=1}^{3N}$
which satisfy
 \be\label{eq:restrict}
\xi_{i+N}=\xi_i+1
 \ee
  (addition on the site index is modulo $3N$); that is, if an $A$ particle
is on site $i$ then there must be a $B$ particle at site $i+N$ and a $C$
particle at site $i+2N$, etc.  The dynamics for the constrained ring is
given by a modification of the usual rules of the canonical ABC model on a
ring: exchanges occur simultaneously across three equally spaced, unmarked
bonds in the usual ABC manner, with rate 1 for the favored exchanges and
rate $q=e^{-\beta/N}$ for the unfavored ones.

We consider now any fixed block of $N$ consecutive sites on the constrained
ring and ask for the induced dynamics on configurations in this block. Two
types of transitions occur: nearest-neighbor exchanges at standard ABC
rates for a system of size $N$ and inverse temperature $\beta$ (i.e., rates
$1$ and $q=e^{-\beta/N}$) and a transition corresponding to an exchange on
the constrained ring across the boundaries of the block.  To understand the
latter, suppose the configuration within the block has the form 
$(\alpha+2)\,\zetau\,(\alpha+1)$, with $\zetau$ any configuration on $N-2$
sites; then \eqref{eq:restrict} implies that the particles immediately to
the left and right of the block are of type $\alpha$, and a transition from
$(\alpha+2)\,\zetau\,(\alpha+1)$ to $\alpha\,\zetau\,\alpha$ occurs at rate
1.  The reverse transition occurs at rate $q$, and no such transition
occurs when the block configuration is $(\alpha+2)\,\zetau\,\alpha$.  Then
using $\lambda_A=\lambda_B= \lambda_C$ one checks, just as for the dynamics
considered above, that if one identifies the block with an interval of $N$
sites then this dynamics satisfies the detailed balance condition with
respect to the grand canonical Gibbs measure \eqref{eq:mu}.

On the constrained ring there are equal numbers of $A$, $B$, and $C$
particles, from \eqref{eq:restrict}, so that the energy $E_{3N}$ (that is,
the energy given by \eqref{eq:EN} with $N$ replaced by $3N$ throughout),
and thus the restriction of the Gibbs measure $Z^{-1}\exp\{-\beta E_{3N}\}$
to particle configurations satisfying \eqref{eq:restrict}, is well defined
and independent of the starting point of the summations \cite{Evans98}.
Moreover, this is the invariant measure for the constrained ring dynamics
defined above, as one again checks by verifying detailed balance. With the
discussion above this shows that the restriction of
$Z^{-1}\exp\{-\beta E_{3N}\}$ to the block of $N$ sites is the Gibbs
measure \eqref{eq:mu}.  One may also verify this from the fact that if
$\xiu$ is a constrained ring configuration and $\zetau$ the portion of that
configuration within the block then
 \be\label{EN3N}
E_{3N}(\xiu)=E_N(\zetau)+N/3.
 \ee
 Thus we can study the GCABC with $\lambda_A=\lambda_B= \lambda_C$ by
studying directly the constrained ring.

\subsection{The scaling limit for the constrained ring\label{sec:sclim}}

To identify typical coarse-grained density profiles at large $N$ on the
constrained ring we consider the scaling limit \eqref{eq:scale} with $N$
replaced by $3N$ ($N\to\infty$ with $i/3N\to x\in[0,1]$) and find the
appropriate free energy functional.  The scaling limit of the energy per
site is still given by \eqref{eq:E}, but because the full microscopic
configuration under the constraint \eqref{eq:restrict} is determined by the
configuration of the first $N$ sites the entropy per unit site is only 1/3
of \eqref{eq:S}.  This leads to a free energy functional
 \be\label{eq:3F}
\beta\E(\{n\})-\frac13\S(\{n\})=\frac13\F^{(3\beta)}(\{n\}).
 \ee
 Here $\E(\{n\})$ and $\S(\{n\})$ are as in \eqref{eq:E} and \eqref{eq:S}
and $n$ is a constrained density profile, that is, one which
satisfies the continuum equivalent of \eqref{eq:restrict}: 
 \be\label{eq:image}
n_\alpha(x)=n_{\alpha+1}(x+1/3),
 \ee
 where the addition $x+1/3$ is taken modulo 1.    $\F^{(3\beta)}$ is the
free energy functional at temperature $3\beta$ of the (unconstrained)
canonical system on an interval, as defined in \eqref{eq:S}; equivalently,
because there are equal numbers of particles of each species, this is the
free energy functional on a ring \cite{ACLMMS}.

Typical (coarse-grained) profiles at inverse temperature $\beta$ on the
constrained ring, for large $N$, correspond then to continuum density
profiles $\rho(x)$ which satisfy the constraint \eqref{eq:image} and
minimize the free energy over all such constrained profiles.  It follows
from \eqref{eq:solns} and \eqref{eq:phases}, however, that the typical
profiles (minimizers) for the canonical free energy, which are a priori
unconstrained, do in fact satisfy \eqref{eq:image}.  Thus by \eqref{eq:3F}
the typical profiles for the constrained ring are the same as the typical
profiles of an unconstrained system on the ring at inverse temperature
$3\beta$.  This is illustrated in Figure~\ref{fig:profiles}, where we plot
time-averaged profiles from Monte-Carlo simulations of the constrained ring
at $\beta=10.152$ and the exact solution \cite{ACLMMS} for the
unconstrained ring at $\beta=30.456$, showing close agreement.  (We can use
time averaging rather than spatial coarse graining for this comparison
because the time scale for the profile to drift around the ring is much
larger than the simulation time scale.)

\begin{figure}
\centerline{\includegraphics[width=12cm,height=6cm]{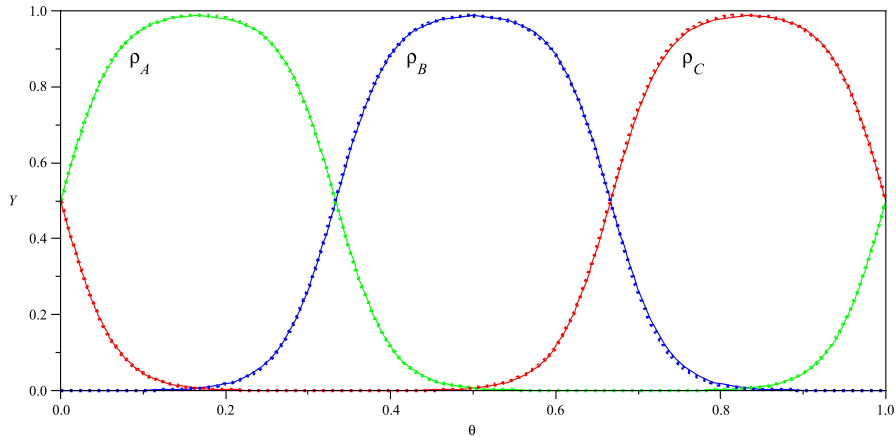}}
\caption{Typical profiles in a  large system.  The dotted curves are
  time-averaged occupation numbers in a constrained ring of size $N=1800$
  at inverse temperature $\beta=10.152$.
  The solid curves are the corresponding elliptic functions obtained from
  the exact solution of \cite{ACLMMS} at temperature $\beta=30.456$.}
\label{fig:profiles}
\end{figure}

It follows from this discussion that when the chemical potentials are equal
the critical temperature $\hat\beta_c$ for the grand canonical ensemble on
an interval, which is represented by the part of the constrained ring
between two markers, is $\hat\beta_c=\beta_c/3$.  Typical configurations
are constant if $\beta<\hat\beta_c$ and for $\beta>\hat\beta_c$ are a
portion of the typical profile for the canonical system at inverse
temperature $3\beta$; the latter is periodic and in the GCABC system we see
a randomly-selected one-third of a period.  These properties are confirmed
in Section~\ref{sec:scaling} by direct analysis of the grand canonical
system in the scaling limit.

\section{The  phase diagram of the GCABC model\label{sec:scaling}}

In this section we discuss the GCABC model directly in the scaling limit
\eqref{eq:scale}.  From \eqref{eq:mu} we see that the new free energy
functional $\Fh(\{n\})\,\bigl(=\Fh_{\beta,\lambda}(\{n\})\bigr)$, which is
the negative of the pressure multiplied by $\beta$, is obtained by adding
chemical potential terms to the free energy functional of the canonical
model:
 \be
\Fh(\{n\})
  =\F(\{n\})-\sum_\alpha\lambda_\alpha\int_0^1dx\,n_\alpha(x), 
 \ee
 with $\F$ given by \eqref{eq:F}.  The profiles now are constrained
 only by 
 \be   \label{eq:gcconstraints}
0\le n_\alpha(x)\le1   \quad\hbox{and}\quad \sum_\alpha n_\alpha(x) = 1.
 \ee
We will always normalize the chemical potentials so that
$\sum_\alpha\lambda_\alpha=0$ (although with this normalization
we cannot conveniently consider the limit in which just one of the
$\lambda_\alpha$ becomes infinite).  Just as for the canonical model
\cite{ACLMMS} it can be shown on general grounds that for every
$\beta,\lambda$ the free energy functional has at least one minimizing
profile $\rho(x)$ which belongs to the interior of the constraint region,
i.e., satisfies $0<\rho_\alpha(x)<1$ for all $\alpha,x$ (and of course
$\sum_\alpha\rho_\alpha(x)=1$ for all $x$).  From this it follows that
$\rho(x)$ will satisfy
 \begin{align}
\frac{\delta}{\delta \rho_A(x)}
    \left[\Fh(\{\rho\})\Big|_{\rho_C=1-\rho_A-\rho_B}\right]
    &=(\F_A(x)-\lambda_A)-(\F_C(x)-\lambda_C)=0,\\
\frac{\delta}{\delta \rho_B(x)}
    \left[\Fh(\{\rho\})\Big|_{\rho_C=1-\rho_A-\rho_B}\right]
    &=(\F_B(x)-\lambda_B)-(\F_C(x)-\lambda_C)=0,
 \end{align}
with $\F_\alpha(x) $ as in \eqref{eq:FA}, so that $\F_\alpha(x)-\lambda_\alpha$
is independent of $\alpha$.  But one finds from \eqref{eq:FA} that
$\sum_\alpha\rho_\alpha\partial\F_\alpha(x)/\partial x=0$, so that
 \be\label{eq:C}
\F_\alpha(x)-\lambda_\alpha = C
 \ee
 for some $C$ independent of $x$ and $\alpha$.  Differentiating
\eqref{eq:C} leads again to \eqref{eq:ABC}:
  \be\label{eq:ABC2}
 \frac{d\rho_\alpha}{dx} =  
  \beta\rho_\alpha ( \rho_{\alpha-1}- \rho_{\alpha+1}),\qquad \alpha=A,B,C.
 \ee
 Moreover, \eqref{eq:C} implies that
$\F_\alpha(0)-\lambda_\alpha
    =\F_{\alpha+1}(1)-\lambda_{\alpha+1}$,
 which with \eqref{eq:FA} yields the boundary condition 
 \be \label{eq:BC}
\rho_\alpha(0)e^{-\lambda_\alpha}=\rho_{\alpha+1}(1)e^{-\lambda_{\alpha+1}},
    \qquad \alpha=A,B,C.
 \ee
 Note that \eqref{eq:BC} is consistent with the (first) dynamics described in
 Section~{sec:dynamics}.

 Equations \eqref{eq:ABC2} and \eqref{eq:BC} may be taken as the ELE of the
model (it is easy to verify that these imply \eqref{eq:C}).  Solutions of
\eqref{eq:ABC2} are, by the analysis of \cite{ACLMMS}, of the form
\eqref{eq:solns}, with phase shifts satisfying \eqref{eq:phases}.  It
remains only to consider the effect of the boundary condition
\eqref{eq:BC}.  

Let us begin by considering the case $\lambda_A=\lambda_B=\lambda_C$, in
which \eqref{eq:BC} becomes $\rho_\alpha(0)=\rho_{\alpha+1}(1)$.  Certainly
the constant profile with $\rho_\alpha(x)=1/3$ for all $\alpha,x$ satisfies
this condition and hence is a solution for all $\beta$.  From
\eqref{eq:phases} we see that a nonconstant solution \eqref{eq:solns} will
satisfy this condition if and only if
 \be\label{eq:eq1}
y_K(t_\alpha-\beta)=y_K(t_\alpha+\beta-\tau_K/3), \qquad \alpha=A,B,C.
 \ee
 The properties of $y_K$ mentioned in \crem{rem:yk} imply that
\eqref{eq:eq1} can hold if and only if
$(t_\alpha-\beta)\pm(t_\alpha+\beta-\tau_K/3)$ is an integer multiple of
$\tau_K$.  The choice of the positive sign here leads to no solutions
consistent with \eqref{eq:phases}; the negative sign gives
$2\beta=(3n-2)\tau_K/3$ for $n=1,2,3,\ldots$.  Since the minimal period of
solutions of \eqref{eq:ABC} is $2\beta_c=4\pi\sqrt3$, a nonconstant
solution of \eqref{eq:ABC} and \eqref{eq:BC} can exist only if
$\beta>\beta_c/3$; thus as in Section~\ref{sec:dynamics} we find that
$\hat\beta_c=\beta_c/3$ is the critical inverse temperature for the GCABC
model.  There is no constraint on the $t_\alpha$ other than
\eqref{eq:phases}, so that there is a one-parameter family of solutions
differing by translation.

Following the usage of \cite{ACLMMS} it is natural to refer to the
solutions just discussed for which $2\beta=(3n-2)\tau_K/3$ as being of {\it
type $n$}.  We will, again as in \cite{ACLMMS}, extend this classification
to the case of general $\lambda$: a solution \eqref{eq:solns} of
\eqref{eq:ABC2} and \eqref{eq:BC} will be said to be {\it of type $1$} if
$2\beta\le\tau_K/3$ and {\it of type $n$}, $n=2,3,\ldots$, if
$(3n-5)\tau_K/3<2\beta\le(3n-2)\tau_K/3$. With this terminology we can
state our main result.

 \begin{thm}\label{thm:main} (a) If $\lambda_A=\lambda_B=\lambda_C$ then
for the equations \eqref{eq:ABC2} and \eqref{eq:BC} there exist (i)~the
constant solution, (ii)~for $\beta>(n-2/3)\beta_c=2\pi (n-2/3)\sqrt3$,
$n=1,2,3,\ldots$, a family of solutions of type $n$, differing by
translation, and (iii)~no other solutions.  The minimizers of the free
energy functional $\Fh$ are, for $\beta\le\beta_c/3$, the
(unique) constant solution and, for $\beta>\beta_c$, any type~1 solution.

 \smallskip\noindent
 (b) If not all $\lambda_\alpha$ are equal then there exists for all
$\beta$ a unique minimizer of the free energy functional
$\Fh$; moreover, this minimizer is a type~1 solution of
\eqref{eq:ABC2} and \eqref{eq:BC}.  \end{thm}

We give the proof of part (a) of this theorem here; the more technical
proof of (b) is presented in Appendix~\ref{sec:proof}.  

\begin{proofof}{\cthm{thm:main}(a)} The discussion at the beginning of this
section establishes the first statement of the theorem; it remains to show
that the type~1 solution, when it exists, minimizes the free energy.  We do
so by reducing this problem to the corresponding one for the canonical
ensemble; the argument is similar to the consideration of the constrained
ring system in Section~\ref{sec:dynamics}.  For any profile
$n(x)=(n_A(x),n_B(x),n_C(x))$ (where it is understood that
$0\le n_\alpha(x)\le 1$ and $\sum_\alpha n_\alpha(x)=1$) define the {\it
tripled} profile $\Theta(\{n\})$ by
 \be
(\Theta(\{n\}))_\alpha(x)=\begin{cases}
       n_\alpha(3x),&\hbox{if $0\le x<1/3$},\\
       n_{\alpha-1}(3x-1),&\hbox{if $1/3\le x<2/3$},\\
       n_{\alpha-2}(3x-2),&\hbox{if $2/3\le x<1$}.
  \end{cases} \ee
 The profiles which have the form $\Theta(\{n\})$ for some $n$ are
precisely those satisfying \eqref{eq:image}; in particular, each
$\Theta(\{n\})$ gives equal mean densities to the three species.  

Now a simple computation shows that for any profile $\{n_\alpha(x)\}$,
 \be\label{fetrip}
\F^{(\beta)}(\{n\})=\F^{(3\beta)}(\{\Theta(\{n\})\})-\beta/3.
 \ee
 (Note that this free energy differs by an overall factor, plus an additive
constant, from that of \eqref{eq:3F}; the difference arises because here we
started from the energy and entropy per site on the interval of size $N$,
and in \eqref{eq:3F} from the energy and entropy per site on the ring of
size $3N$.)  Thus the problem of finding the minimizer(s) of
$\Fh^{(\beta,0)}(\{n\})=\F^{(\beta)}(\{n\})$ over all profiles $n$ is
equivalent to finding the minimizer(s) of $\F^{(3\beta)}(\{n\})$ over all
profiles satisfying \eqref{eq:image}.  On the other hand, the minimizers of
$\F^{(3\beta)}$ over all equal-density profiles are given in
\cthm{thm:oldmain}(a): the constant solution if $3\beta\le\beta_c$ and the
solution of (minimal) period $6\beta$ if $3\beta>\beta_c$ (this is the
type~1 solution for the canonical model).  Because these are either
constant or periodic, they satisfy \eqref{eq:image} and hence are also the
minimizers over all such profiles.  But these minimizers are precisely the
images under $\Theta$ of the profiles identified as minimizers in
\cthm{thm:main}(a).  \end{proofof}


\begin{rem} In the argument above the essential role of the tripling map
$\Theta$ is to convert the problem of minimizing $\Fh$ with respect to
arbitrary variations in the profiles to the previously solved problem of
minimizing under variations which preserve the condition
$\int_0^1dx\,n_\alpha(x)=1/3$.  Other conclusions may be obtained
similarly; we mention briefly two examples.

 \smallskip\noindent
  (a) It was shown in \cite{ACLMMS} that, for $\beta<(2/3\sqrt3)\beta_c$
and any $r=(r_A,r_B,r_C)$, $\F(\{n\})$ is convex as a functional of
profiles satisfying \eqref{eq:constraints}. Via $\Theta$ this implies that
$\Fh(\{n\})$ is, for $\beta<(2/3\sqrt3)\hat\beta_c$, convex as a
function of profiles satisfying \eqref{eq:gcconstraints}.

 \smallskip\noindent
 (b) The two point correlation functions on the interval are related to
those on the constrained ring by
 \be\label{eq:icr}
\langle n(x)n(y)\rangle_{\rm interval}
 = \langle n(x/3)n(y/3)\rangle_{\rm ring}.
 \ee
 The latter (denoted below simply as $\langle\cdot\rangle$) may be computed
in the high temperature phase by a calculation parallel to that of
\cite{BDLvW}.  On the constrained ring a perturbation
$n_\alpha(x)=1/3+a_\alpha\cos(2\pi mx)+b_\alpha\sin(2\pi mx)$ of the
constant solution satisfies \eqref{eq:image} and
$\sum_\alpha n_\alpha(x)=1$ if and only if $m=3k+j$ for $j=1$ or $2$, and
 \be
a_{\alpha+1}=-\frac12a_\alpha+(-1)^j\frac{\sqrt3}2b_\alpha,\qquad
b_{\alpha+1}=-\frac12b_\alpha-(-1)^j\frac{\sqrt3}2a_\alpha.
 \ee
 One may thus treat $a_A$ and $b_A$ as the independent parameters.  The
probability of the profile $\{n_\alpha(x)\}$ is
$\exp\{-3N\F^{(3\beta)}(\{n\})\}$, and to quadratic order in the
perturbation,
 \be
 \F^{(3\beta)}(\{n\})\simeq{\rm constant}
  +\frac9{8\pi m} \,\bigl[2\pi m+(-1)^j\sqrt3\beta\bigr]
    \,\bigl(a_A^2+b_A^2\bigr).
 \ee
 Thus
 \be
\langle a_A^2\rangle=\langle b_A^2\rangle
  =\frac{4\pi m}{27N(2\pi m+(-1)^j\sqrt3\beta)},\qquad
 \langle a_Ab_A\rangle=0.
 \ee
  Summing over all the fluctuations, i.e., over $m$, we obtain
 \begin{eqnarray}\label{eq:twopoint}
\langle n_\alpha(x)n_\alpha(y)\rangle_c
  &=& \frac{4\pi}{27N}\sum_{k=0}^\infty\sum_{j=1}^2
\frac{m\cos[2\pi m(x-y)]}{2\pi m+(-1)^j\sqrt3\beta}\,
 \Bigg|_{m=3k+j}.
 \end{eqnarray}
 All connected two-point functions
$\langle n_\alpha(x)n_\gamma(y)\rangle_c$ may be obtained on the
constrained ring from \eqref{eq:twopoint} via \eqref{eq:image}, and then on
the interval using \eqref{eq:icr}.  Note that \eqref{eq:twopoint} diverges
as $\beta\nearrow\hat\beta_c$.

 \end{rem}

\subsection {The canonical free energy $F(r)$ \label{sec:cfe}}

The free energy in the canonical model, for mean densities $r_A$, $r_B$,
$r_C$ satisfying $r_A+r_B+r_C=1$, is given by
 \be\label{eq:cfe}
F(r)=F(r_A,r_B,r_C)=\min_{\{n(x)\}} \F(\{n(x)\}),
 \ee
 with the minimum taken over all profiles $n(x)$ satisfying the constraints
\eqref{eq:constraints}.  The grand canonical free energy may then be
computed in two ways:
 \begin{eqnarray}\label{eq:gcfe}
 \hat F(\lambda)
  &=&\inf_{\{n(x)\}}\Fh(\{n\})\\
  &=& \inf_{\textstyle\atop{\sum_\alpha r_\alpha=1}{r_\alpha\ge0}} 
  \left\{F(r)-\sum_\alpha\lambda_\alpha r_\alpha\right\},
  \label{eq:gcfe2}
 \end{eqnarray}
 where the infimum in \eqref{eq:gcfe} is over all profiles.  We can obtain
information on the structure of $F(r)$ from the above results for the
minimization problem \eqref{eq:gcfe}, together with the trivial remarks
that a unique minimum for \eqref{eq:gcfe} implies a unique minimum for
\eqref{eq:gcfe2} and that such a unique minimum implies that the surface
$y=F(r)$ lies above the plane
$y=\hat F(\lambda)+\sum_\alpha\lambda_\alpha r_\alpha$ and touches it at a
single point.

In particular, the fact that when $\beta\le\hat\beta_c$ there is for all
$\lambda$ a unique minimizer for \eqref{eq:gcfe} implies that for such
$\beta$ the function $F(r)$ is convex.  When $\beta>\hat\beta_c$ the
minimizer for \eqref{eq:gcfe} is unique except in the case
$\lambda_A=\lambda_B=\lambda_C=0$, when the plane mentioned above is
horizontal.  In that case the minimum occurs at points lying above a
certain simple closed curve $\Gamma\,(=\Gamma_\beta)$ in the plane
$\sum_\alpha r_\alpha=1$, with the point $r_A=r_B=r_C=1/3$ in its interior;
sample curves are shown in Figure~\ref{fig:ccurves}.  $\Gamma$ may be
parametrized as $r^*(t)$, $0\le t\le\tau_K$, where $K$ is the parameter in
the type 1 solution of \cthm{thm:main}(a) and
 \be
r_\alpha^*(t)=\frac3{\tau_K}\int_{-\tau_K/6}^{\tau_K/6}y_K(s+t_\alpha+t)\,ds.
 \ee
 (The fact that this curve is simple follows, for example, from
\cprop{recall}(d).)  The three-fold symmetry then implies that the surface
$y=F(r)$ has a ``tricorn'' shape. 

\begin{figure}
\centerline{\includegraphics[width=6cm,height=6cm]{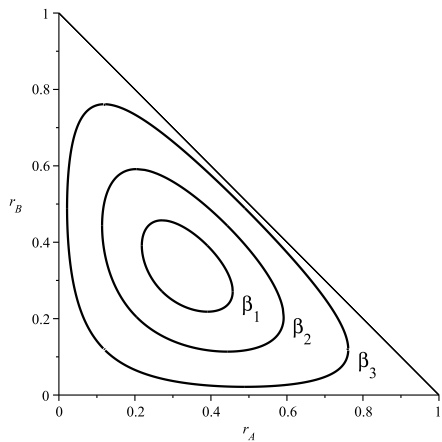}}
\caption{Curves $\Gamma_\beta$ in the $r_A$-$r_B$ plane along which $F(r)$
achieves its minimum value, for $\beta_1=3.75$, $\beta_2=4.25$, and
$\beta_3=6.05$ ($\hat\beta_c=2\pi/\sqrt3\simeq 3.63$).} \label{fig:ccurves}
\end{figure}

\section{Concluding remarks}\label{sec:conclude}

It is natural to compare the phase diagram obtained here for the one
dimensional reflection asymmetric ABC model with that of the corresponding
symmetric model, that is, the mean field three state Potts model (see
\cite{SF},\cite{MS}).  We will define the latter by replacing the sum over
$j>i$ in \eqref{eq:EN} by a sum over all $j\ne i$ and dividing by 2; this
yields
 \be\label{eq:symEN}
 E_N^*(\zetau)=\frac1{2N}[N_AN_C+N_CN_B+N_BN_A]
   =\frac1{4N}[N^2-(N_A^2+N_B^2+N_C^2)].
 \ee
 The energy \eqref{eq:symEN} is related to that of the standard mean field
 Potts model \cite{MS} by a choice of energy scale and a shift of the
 ground state energy.  It is, as is usual for mean field models,
 independent of dimension and geometry.  
 There is thus no spatial structure in the system and the
 canonical measure just gives equal weights to all configurations.

The canonical free energy functional with prescribed values of
$r_\alpha=\int_0^1n_\alpha(x)\,dx$ is
 \be
  \F^*(\{n\})=\frac\beta2[r_Ar_C+r_Br_A+r_Cr_B]-\S(\{n\}),
 \ee
 with $\S\left(\{n\}\right)$ still given by \eqref{eq:S}.  For all $\beta$
the minimizers of $\F^*$ are the constant density profiles
$\rho_\alpha(x)=r_\alpha$, and there are no phase transitions of any kind
in the canonical system.  The corresponding minimum value
 \be
 F^*(r)=\frac\beta2\sum_\alpha r_\alpha r_{\alpha+2}
   -\sum_\alpha r_\alpha\log r_\alpha
  =\frac\beta4\sum_\alpha r_\alpha^2 +\frac\beta4
   -\sum_\alpha r_\alpha\log r_\alpha
 \ee
 of $\F^*(\{n\})$ is in fact just the value of $\F(\{n\})$ evaluated at
 these constant profiles (this follows from our choice of the factor $1/2$
 in \eqref{eq:symEN}), and hence is an upper bound for the free energy $F(r)$ of
 \eqref{eq:cfe}.

The situation is quite different for the grand canonical ensemble.  Here
the analogue of \eqref{eq:gcfe} is
 \be\label{eq:symgcfe}
 \hat F^*(\lambda)=\inf_r\,
   \Bigl\{F^*(r)-\sum_\alpha\lambda_\alpha r_\alpha\Bigr\}
 \ee
 The analysis of $F^*(\lambda)$ leads to a phase diagram completely
 different from that of the reflection asymmetric grand canonical model
 considered in Sections~3 and 4 above \cite{SF}.  In particular, \eqref{eq:symgcfe}
 exhibits a first order phase transition for
 $\lambda_A=\lambda_B=\lambda_C$ at $\beta_c^*=8\log2$.  For
 $\beta<\beta_c^*$ the minimizer is $r_A=r_B=r_C=1/3$; for
 $\beta>\beta_c^*$ there are three minimizers, each rich in one of the
 three species, and at $\beta=\beta_c^*$ all four of these states are
 minimizers.

\subsection{Higher dimensions\label{sec:hd}}

As was already noted and is well known, the standard mean field models with
symmetric interactions do not depend on the dimension or topology of the
spatial structure of the system considered.  This is clearly not the case
for models with reflection asymmetric interactions, such as the
one-dimensional ABC model considered in this paper.

We comment now on various possible generalizations of such reflection
asymmetric mean field models to higher dimensions, taking for simplicity
the dimension to be two and the lattice to be an $N\times N$ square in
$\bZ^2$.  Let us consider first a situation in which the mean field
interactions are symmetric in the vertical direction but of the form
\eqref{eq:EN} in the horizontal direction. This yields an energy of the
form
\begin{eqnarray} \label{eq:ENbar}
\tilde E(\zetau)  &=& {\frac{1}{N^2}} \sum_{k,l}^N
\sum_{i=1}^{N-1}\sum_{j=i+1}^N\sum_\alpha
   \eta_\alpha (i,k) \eta_{\alpha+2}(j,l)\\
  &=&\sum_{i=1}^{N-1}\sum_{j=i+1}^N\sum_\alpha
   \tilde\eta_\alpha(i)\tilde\eta_{\alpha+2}(j),\label{eq:ENbar2}
\end{eqnarray}
 where
 \be\label{eq:etatilde}
\tilde\eta_\alpha(i)=\frac1N\sum_{k=1}^N\eta_\alpha(i,k).
 \ee
 The energy
functional $\tilde\E$ obtained from \eqref{eq:ENbar2} in the scaling limit
is identical to that given in \eqref{eq:E} with $n_\alpha(x)$ replaced by
$\tilde n_\alpha(x)=\int_0^1n(x,y)\,dy$.  The entropy term (compare
\eqref{eq:S}),
\be
 -\tilde \S=\sum_\alpha\int_0^1\int_0^1
    n_\alpha(x,y)\log n_\alpha(x,y)\,dx\,dy,
 \ee
 is clearly minimized, subject to a specified $\{\tilde n_\alpha(x)\}$, by
setting $n_\alpha(x,y)=\tilde n_\alpha(x)$, and so density profiles which
minimize $\beta\tilde\E-\tilde\S$ depend only on $x$ and are the same as
for the one dimensional case, both for the canonical and grand canonical
ensembles.

\begin{rem} \label{rem:chiral} The two-dimensional chiral clock model
\cite{O,H,AP} also contains interactions---in that case, nearest-neighbor
ones---which are reflection symmetric in the vertical direction but not in
the horizontal one.  When the parameter $\Delta$ (in the notation of
\cite{O}) has value $1/2$ the energy, up to an additive constant and a
rescaling, is
 \be
\sum_{i=1}^{N-1}\sum_{k=1}^N\sum_\alpha
   \eta_\alpha (i,k) \eta_{\alpha+2}(i+1,k)
    - \sum_{i=1}^N\sum_{k=1}^{N-1}\sum_\alpha
   \eta_\alpha (i,k) \eta_\alpha(i,k+1)
  ,\label{eq:ENchiral}
 \ee
 so that the interactions in the horizontal direction have a form
 reminiscent of \eqref{eq:EN}. 
\end{rem}

A second possibility is to take  the reflection asymmetry to be  the same in
the $x$ and $y$ directions.  In this case \eqref{eq:EN} takes the form
\be
E_{N^2}(\zetau)  = {\frac{1}{N^2}} \sum_\alpha\sum^{N-1}_{i,k=1}
\sum^{N}_{\atop{j=i+1}{l=k+1} }
   \eta_\alpha (i,k) \eta_{\alpha+2}(j,l). \label{eq:ENtilde}
\ee
The analysis of this model seems considerably more complicated and we will
attempt no discussion here.

\medskip
\noindent {\bf Acknowledgments:} We thank Lorenzo Bertini, Thierry
Bodineau, Eric Carlen, Or Cohen, Bernard Derrida, and David Mukamel for
helpful discussions.  The work of J.B.~and J.L.L.~was supported by NSF
Grant DMR-0442066 and AFOSR Grant AF-FA9550-04.

\appendix

\section{Proof of \cthm{thm:main}(b)}\label{sec:proof}

We begin by giving an alternate form of the boundary conditions
\eqref{eq:BC}.  With \eqref{eq:solns} and \eqref{eq:phases} these become
 \be
\lambda_\alpha-\lambda_{\alpha+1}
   =\log\bigl(y_K(t_\alpha-\beta)\bigr)
      -\log\bigl(y_K(t_\alpha+\beta-\tau_K/3)\bigr).
 \ee
 From \eqref{eq:ABC} and \eqref{eq:solns},
$(\log y_K(t))'=[y_K(t+\tau_K/3)-y_K(t-\tau_K/3)]/2$, so that 
 \be\label{eq:newbc}
\lambda_\alpha-\lambda_{\alpha+1}
   =\frac12\int_{t_\alpha+\beta-\tau_K/3}^{t_\alpha-\beta}
  \left[y_K\left(t+\frac{\tau_K}3\right)
    -y_K\left(t-\frac{\tau_K}3\right)\right]\,dt.
 \ee
 The solution of \eqref{eq:newbc} which also satisfies
   $\sum_\alpha\lambda_\alpha=0$ is
 \be\label{eq:newbca}
\lambda_\alpha=
  \frac12\int_{s_\alpha-(\tau_K/6-\beta)}^{s_\alpha+(\tau_K/6-\beta)}
   \left(\frac13-y_K(t)\right)\,dt,
 \ee
 where $s_\alpha=t_\alpha+\tau_K/2$.  The form \eqref{eq:newbca} is
 convenient when $2\beta\le\tau_K/3$; if  $2\beta\ge\tau_K/3$ we may
 rewrite this as
 \be\label{eq:newbcb}
\lambda_\alpha=
  \frac12\int_{s_\alpha-(\beta-\tau_K/6)}^{s_\alpha+(\beta-\tau_K/6)}
   \left(y_K(t)-\frac13\right)\,dt.
 \ee
 The representations \eqref{eq:newbca} and \eqref{eq:newbcb} are useful
because they translate the boundary conditions for the grand canonical
model into a form similar to the condition \eqref{eq:ra} in the canonical
model.

We need also to recall from \cite{ACLMMS} some further properties of the
function $y_K(t)$ and its definite integrals
 \be\label{WD}
 Y(K,s,\delta)=\int_{s-\delta}^{s+\delta}y_K(t)\,dt \quad\hbox{and}\quad
 W(K,s,\delta)=\int_{s-\delta}^{s+\delta}y_K\left(t+\frac{\tau_K}3\right)\,dt.
 \ee
 Note that from  \eqref{eq:ra}, 
 \be\label{eq:rax}
r_\alpha=\frac1{2\beta}Y(K,t_\alpha,\beta) 
 \ee
 and that from \eqref{eq:phases},
 \begin{eqnarray}\label{eq:lax1}
 \lambda_\alpha
  &=&\frac{\delta}{3}-\frac12 Y(K,s_\alpha,\delta)
  =\frac{\delta}{3}-\frac12 W(K,s_{\alpha+1},\delta),
 \end{eqnarray}
 for $\delta=\tau_K/6-\beta\ge0$, while for $\delta'=-\delta>0$,
 \begin{eqnarray}\label{eq:lax2}
 \lambda_\alpha
  &=&\frac12 Y(K,s_\alpha,\delta')-\frac{\delta'}{3}
  =\frac12 W(K,s_{\alpha+1},\delta')-\frac{\delta'}{3}.
 \end{eqnarray}

\begin{prop}\label{recall} For $0<K<1/27$:

\noindent
 (a) (i) $y_K(t)$ is even and $\tau_K$-periodic (and hence also symmetric
about $t=\tau_K/2$), takes its minimum value at $t=0$, is strictly
increasing on $[0,\tau_K/2]$, and takes its maximum value at $t=\tau_K/2$.
Moreover, (ii)~$y_K(t-\tau_K/3)+y_K(t)+y_K(t+\tau_k/3)=1$ for all $t$.

 \smallskip\noindent
 (b) The minimum value $a=a(K)=y_K(0)$ of $y_K$ is an increasing function
of $K$ satisfying $0<a(K)<1/3$.  The maximum value $b=b(K)=y_K(\tau_K/2)$
is $b=[2-a-\sqrt{4a-3a^2}]/2$, and $y_K(\tau/6)=(1-b)/2$,
$y_K(\tau/3)=(1-a)/2$.

 \smallskip\noindent
 (c) (i) For fixed $K$ and $\delta$, with $0<\delta<\tau_K/2$, the function
$Y(K,t,\delta)$ shares with $y_K(t)$ the properties listed in (a.i).
Moreover, (ii)
 \be
Y(K,t-\tau_K/3,\delta)+Y(K,t,\delta)+Y(K,t+\tau_k/3,\delta)=2\delta.
 \ee

 \smallskip\noindent
  (d) For $0<\delta<\tau_k/2$, $Y(K,0,\delta)$ is strictly decreasing, and
$W(K,\tau_k/6,\delta)$ strictly increasing, in $K$.

 \smallskip\noindent
Finally, for $0<K_2<K_1<1/27$:

\noindent (e) (i) For any $t_0$ the curves $y_{K_1}(t_0+t)$ and
$y_{K_2}(t)$ intersect exactly once in the interval
$0\le t\le \tau_{K_2}/2$, and (ii)~$y_{K_2}(t)<y_{K_1}(t)$ and
$y_{K_2}(\tau_{K_2}/2-t)>y_{K_1}(\tau_{K_1}/2-t)$ for
$0\le t\le \tau_{K_1}/6$.  \end{prop}

 \smallskip\noindent
 \begin{proof} These results either appear in \cite{ACLMMS} or are
immediate consequences of results appearing there.  For (a) and (b) see
Section~5.2 of \cite{ACLMMS} and in particular Remark~5.1(a); for (c.i) see
Remark~5.3(b). (c.ii) follows from (a.ii).  The first statement of (d)
follows from the fact that $Y(K,0,\delta)$ is continuous in $K$ and, for
$0<\delta<\tau_k/2$, approaches $2\delta/3$ as $K\nearrow1/27$ and $0$ as
$K\searrow0$, together with Theorem 6.1 of \cite{ACLMMS} which, if one
takes there $r_A=r_C$, asserts that for given $r_B$ with $0<r_B<1/3$ there
is at most one value of $K$ satisfying \eqref{eq:rax}.  The second
statement of (d) is verified similarly.  Finally, (e.i) is a special case
of Lemma~6.2(a) of \cite{ACLMMS} and (e.ii) then follows from (e.i) and the
inequalities $y_{K_2}(0)<y_{K_1}(0)$,
$y_{K_2}(\tau_{K_1}/6)< y_{K_2}(\tau_{K_2}/6)< y_{K_1}(\tau_{K_1}/6)$,
$y_{K_2}(\tau_{K_2}/2)>y_{K_1}(\tau_{K_1}/2)$, and
$y_{K_2}(\tau_{K_2}/2-\tau_{K_1}/6)>y_{K_2}(\tau_{K_2}/3)
>y_{K_1}(\tau_{K_1}/3)$,
easily obtained from the properties given in (a) and (b).  \end{proof}

We now turn to the proof of \cthm{thm:main}(b).  We know (see the remarks
at the beginning of Section~\ref{sec:scaling}) that there is at least one
minimizer and that every minimizer satisfies the ELE \eqref{eq:ABC2},
\eqref{eq:BC}.  Thus the conclusion of the theorem will follow
from:

\begin{lem}\label{lem:mainl} If $\lambda_A$, $\lambda_B$, and $\lambda_C$
  are not all equal then:

 \smallskip\noindent
 (a) No solution of \eqref{eq:ABC2}, \eqref{eq:BC} of type $n$, $n\ge2$,
can minimize $\Fh$.

 \smallskip\noindent
 (b) At most one solution of \eqref{eq:ABC2}, \eqref{eq:BC} of type 1
exists.  \end{lem}

\begin{rem}\label{rem:order} In proving \clem{lem:mainl} we need not
consider either the constant solution of the ELE or nonconstant solutions
for which $2\beta=(n-2)\beta_c/3$, both of which satisfy \eqref{eq:BC} only
when all the $\lambda_\alpha$ are equal.  We may also suppose, without loss
of generality, that
 \be\label{eq:lamorder}
 \lambda_A\le\lambda_C\le\lambda_B, 
  \quad\hbox{ with $\lambda_A<\lambda_C$ or $\lambda_C<\lambda_B$.}
 \ee
 If $\delta=\tau_K/6-\beta>0$ it then follows from \eqref{eq:lax1} that
$Y(K,s_B,\delta)\le Y(K,s_C,\delta)\le Y(K,s_A,\delta)$ and then from
Proposition~\ref{recall}(c.i) (see Figure~\ref{fig:Y}, which displays
graphically the qualitative properties of $Y(K,s,\delta)$ implied there)
that $0\le s_B\le\tau_K/6$, so that $\tau_K/2\le t_B\le2\tau_K/3$ and
hence, from Proposition~\ref{recall}(c.i) and \eqref{eq:rax}, that
$r_A\le r_C \le r_B$.  If $\lambda_A<\lambda_C$ then $s_B>0$,
$t_B>\tau_K/2$, and $r_A<r_C$; similarly $r_C<r_B$, if
$\lambda_C<\lambda_B$.  Similarly, if $\tau_K/6<\beta<\tau_K/2$ then (now
using \eqref{eq:lax2}) $\tau_K/2\le s_B\le2\tau_K/3$ and
$r_B\le r_C \le r_A$, again with strict inequality for two of the
$\lambda_\alpha$ implying the corresponding inequality for the
$r_\alpha$.\end{rem}
 
\begin{figure}
\centerline{\includegraphics[width=12cm,height=6cm]{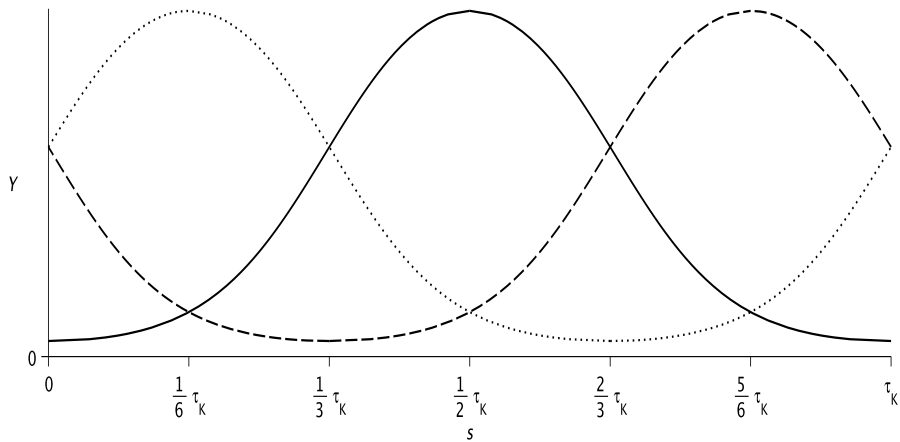}}
\caption{Plots showing qualitative features of $Y(K,s,\delta)$ (solid),
$Y(K,s+\tau_K/3,\delta)$ (dotted), and $Y(K,s-\tau_K/3,\delta)$ (dashed)
for $0<\delta<\tau_K/2$, based on \cprop{recall}(c.1).} 
\label{fig:Y}
\end{figure}

\begin{proofof}{\clem{lem:mainl}(a)} Consider some type $n$ solution
$\rho(x)$, $n\ge2$, of \eqref{eq:ABC2}, \eqref{eq:BC}; $\rho(x)$ has the
form \eqref{eq:solns} with $2\beta>\tau_K/3$.  We need to find a profile
$\tilde\rho(x)$ with $\Fh(\{\tilde\rho\})<\Fh(\{\rho\})$.  There are three
subcases:

 \smallskip\noindent
 {\bf Case (a.i)~$2\beta>\tau_K$.} In this case it was shown in
\cite{ACLMMS} that there is a {\it rearrangement} $\tilde\rho(x)$ of
$\rho(x)$ with $\F(\{\tilde\rho\})<\F(\{\rho\})$.  This
rearrangement does not change the mean densities $r_\alpha$ and hence also
$\Fh(\{\tilde\rho\})<\Fh(\{\rho\})$.
 \smallskip\noindent

 \smallskip\noindent
 {\bf Case (a.ii)~$2\beta=\tau_K$.} In this case the solution $\rho(x)$ has
mean densities $r_\alpha=1/3$, so that $\sum\lambda_\alpha r_\alpha=0$.
From the description of the curve $\Gamma$ in Section~\ref{sec:cfe} it
follows that for some $z>0$ there exists a minimizer $\tilde\rho(x)$ of
$\Fh^{(\beta,0)}$ with mean densities
$\tilde r_\alpha=1/3+z\lambda_\alpha$, so that
$\sum\lambda_\alpha \tilde r_\alpha>0$.  But then
 \be
 \Fh^{(\beta,\lambda)}(\{\rho\})= \Fh^{(\beta,0)}(\{\rho\})
   > \Fh^{(\beta,0)}(\{\tilde\rho\})> \Fh^{(\beta,\lambda)}(\{\tilde\rho\}).
 \ee

 \smallskip\noindent
 {\bf Case (a.iii)~$\tau_K>2\beta>\tau_K/3$.} By \crem{rem:order},
$r_B\le r_C\le r_A$,  with $r_B< r_C$ if $\lambda_C<\lambda_B$
 and $r_C<r_A$ if $\lambda_A<\lambda_C$.  Consider now the
profile $\tilde\rho$ with $\tilde\rho_\alpha(x)=\rho_{\alpha+1}(x)$.  The
canonical free energy functional satisfies
$\F(\{\tilde\rho\})=\F(\{\rho\})$ and so
 \begin{align}\label{eq:lamr}
\Fh(\{\tilde\rho\})-\Fh(\{\rho\})
  &= \sum_\alpha\lambda_\alpha r_\alpha
   -\sum_\alpha\lambda_\alpha r_{\alpha+1}\cr
 &\hskip-40pt
 =(\lambda_A-\lambda_C)(r_A-r_B)+(\lambda_B-\lambda_C)(r_B-r_C)<0.
 \end{align}
\end{proofof}

The next result, the key to the proof of \clem{lem:mainl}(b), gives certain
monotonicity properties of $Y(K,s,\delta)$ and $W(K,s,\delta)$.

\begin{lem}\label{monotone} If $K$, $s$, and $\delta$ satisfy $0<K<1/27$ and
$0\le s,\delta\le\tau_K/6$, then:

 \smallskip\noindent
 (a) For fixed $K$ and $s$ the function $2\delta/3-Y(K,s,\delta)$
 (respectively $2\delta/3-W(K,s,\delta)$) is strictly increasing
 (respectively strictly decreasing) in $\delta$;

 \smallskip\noindent
 (b) For fixed $K$ and $\delta$ the functions $Y(K,s,\delta)$ and
$W(K,s,\delta$) are strictly increasing in $s$;

 \smallskip\noindent
 (c) For fixed $s$ and $\delta$ the function $Y(K,s,\delta)$
 (respectively $W(K,s,\delta)$) is strictly increasing
 (respectively strictly decreasing) in $K$.
\end{lem}

\begin{proof} (a) We rely throughout on \cprop{recall}(a,b).  From
$0\le s+\delta\le\tau_K/3$ and $-\tau_K/6\le s-\delta\le\tau_K/6$ it
follows that $y_K(s+\delta)\le(1-a)/2$ and $y_K(s-\delta)\le(1-b)/2$. Then
from \eqref{WD},
 \be\label{deriv1}
\frac{d}{d\delta}\left[\frac{2\delta}3-Y(K,s,\delta)\right]
    =\frac23-y_K(s+\delta)-y_K(s-\delta)\ge \frac{a+b}2-\frac13>0, 
 \ee
 as is easily verified from $b=[2-a-\sqrt{4a-3a^2}]/2$ with $0<a<1/3$.  To
show that $(d/d\delta)(2\delta/3-W(K,s,\delta))<0$ it suffices similarly to
verify that
 \be\label{deriv2}
z(K,s,\delta):=y_K\left(s+\frac{\tau_K}3+\delta\right) 
    +y_K\left(s+\frac{\tau_K}3-\delta\right)>\frac23. 
 \ee
 Because $y_K$ is even and $\tau_K$-periodic, $z$ is invariant under
$(s,\delta)\to (s',\delta')$ with $s'=\tau_K/6-\delta$,
$\delta'=\tau_K/6-s$, so that it suffices to verify \eqref{deriv2} for
$s+\delta\le\tau_K/6$, and since under this condition both terms in
$z(K,s,\delta)$ are increasing in $s$ it suffices to consider $s=0$.  But
 because $y_K$ is even, 
 \begin{eqnarray}
    z(K,0,\delta)&=&\frac12\Bigl[y_K\left(\frac{\tau_K}3+\delta\right) 
    +y_K\left(-\frac{\tau_K}3+\delta\right)\nonumber\\
  &&\hskip40pt 
     + y_K\left(-\frac{\tau_K}3-\delta\right)
    +y_K\left(\frac{\tau_K}3-\delta\right)\Bigr] \nonumber\\
    &=&1-\frac12\left[y_K(\delta)+y_K(-\delta)\right]\ge\frac{1+b}2>2/3.
 \end{eqnarray}

 \smallskip\noindent
 (b) See \cprop{recall}(c).

 \smallskip\noindent
 (c) The proofs for $Y$ and of $W$ are similar and we check only $Y$.
Suppose that $0<K_2<K_1<1/27$ and that for some $s_*\in[0,\tau_{K_1}/6]$,
 \be\label{eq:ineq3}
Y(K_1,s_*,\delta)\le Y(K_2,s_*,\delta).
 \ee
 Then certainly $y_{K_1}(t_*)<y_{K_2}(t_*)$ for some
$t_*\in[s_*-\delta,s_*+\delta]$, and since $y_{K_1}(0)>y_{K_2}(0)$,
\cprop{recall}(e.i) implies that $y_{K_1}(t)\le y_{K_2}(t)$ for
$t\in[t_*,\tau_{K_1}/3]$.  Then for $s\in[s_*,\tau_{K_1}/6]$,
 \bea
 \frac{d}{ds}\left[Y(K_2,s,\delta)-Y(K_1,s,\delta)\right]\hskip-130pt
    &&\nonumber\\
 &=& \left[y_{K_1}(s-\delta)-y_{K_2}(s-\delta)\right]
    +\left[y_{K_2}(s+\delta)-y_{K_1}(s+\delta)\right]\nonumber\\
   &>& 0,\label{eq:ineq4}
 \eea
 since both terms on the right had side are positive.  But for 
$0\le\delta\le\tau_{K_1}/6$,
 \be\label{eq:ineq1a}
      W(K_1,\tau_{K_1}/6,\delta)<W(K_2,\tau_{K_2}/6,\delta),
 \ee
 by \cprop{recall}(e.ii), and so from \cprop{recall}(a),
 \begin{eqnarray}\label{eq:ineq2}
 Y(K_1,\tau_{K_1}/6,\delta)
  &=&\frac12\int_{\tau_{K_1}/6-\delta}^{\tau_{K_1}/6+\delta}
   \left[y_{K_1}\left(t-\frac{\tau_{K_1}}3\right)+y_{K_1}(t)\right]\,dt
      \nonumber\\
  &=&\frac12\int_{\tau_{K_1}/6-\delta}^{\tau_{K_1}/6+\delta}
   \left[1-y_{K_1}\left(t+\frac{\tau_{K_1}}3\right)\right]\,dt\nonumber\\
  &=&\delta-\frac12 W(K_1,\tau_{K_1}/6,\delta)\nonumber\\
  &>&\delta-\frac12 W(K_2,\tau_{K_2}/6,\delta)\nonumber\\
  &=& Y(K_2,\tau_{K_2}/6,\delta)\nonumber\\
  &>& Y(K_2,\tau_{K_1}/6,\delta) 
 \end{eqnarray}
 since $\tau_{K_1}<\tau_{K_2}$ (see \crem{rem:yk}), contradicting
\eqref{eq:ineq3} and \eqref{eq:ineq4}.  \end{proof}

\begin{proofof}{\clem{lem:mainl}(b)} For type 1 solutions we
have from \eqref{eq:lax1} that 
 \be\label{eq:newbcc}
\lambda_B
  =\frac{\delta}{3}-\frac12 Y(K,s_B,\delta),\qquad
 \lambda_A
  =\frac{\delta}{3}-\frac12 W(K,s_B,\delta),
 \ee
 with $\delta=\tau_K/6-\beta>0$ and, by \crem{rem:order},
$0\le s_B\le \tau_K/6$.  Thus the existence for some $\lambda$ of two
type~1 solutions would correspond to the existence of $(K_1,s_1)$ and
$(K_2,s_2)$ with $0<K_2<K_1<1/27$ and $0\le s_i\le\tau_{K_i}/6$, $i=1,2$,
such that $2\delta_1/3-Y(K_1,s_1,\delta_1)=2\delta_2/3-Y(K_2,s_2,\delta_2)$
and $2\delta_1/3-W(K_1,s_1,\delta_1)=2\delta_2/3-W(K_2,s_2,\delta_2)$,
where  $\delta_i=\tau_{K_i}/6-\beta$ for $i=1,2$.  Then from
\clem{monotone}(a,c),
 \begin{eqnarray}
 \frac{2\delta_2}3-Y(K_2,s_2,\delta_2) 
  &=&  \frac{2\delta_1}3-Y(K_1,s_1,\delta_1)\nonumber\\
  &<&  \frac{2\delta_1}3-Y(K_2,s_1,\delta_1)\nonumber\\
  &<&  \frac{2\delta_2}3-Y(K_2,s_1,\delta_2),
 \end{eqnarray}
 so that  \clem{monotone}(b) implies that  $s_1<s_2$.  But also 
 \begin{eqnarray}
 \frac{2\delta_2}3-W(K_2,s_2,\delta_2) 
  &=&  \frac{2\delta_1}3-W(K_1,s_1,\delta_1)\nonumber\\
  &>&  \frac{2\delta_1}3-W(K_2,s_1,\delta_1)\nonumber\\
  &>&  \frac{2\delta_2}3-W(K_2,s_1,\delta_2),
 \end{eqnarray}
 implying that $s_1>s_2$, a contradiction.
\end{proofof}

\end{document}